%% file: arxiv.tex
\documentclass[sigplan,screen,nonacm]{acmart}
\usepackage{prooftree}
\usepackage{stmaryrd}
\usepackage{mathtools}
\usepackage{fancyvrb}
\usepackage{macros}
\usepackage{fancyvrb}
\usepackage{microtype}
\usepackage{flushend}

\copyrightyear{2024}
\acmYear{2024}
\setcopyright{acmlicensed}\acmConference[Haskell '24]{Proceedings of the 17th ACM SIGPLAN International Symposium on Haskell}{September 6--7, 2024}{Milan, Italy}
\acmBooktitle{Proceedings of the 17th ACM SIGPLAN International Symposium on Haskell (Haskell '24), September 6--7, 2024, Milan, Italy}
\acmDOI{10.1145/3677999.3678274}
\acmISBN{979-8-4007-1102-2/24/09}



\begin{document}

\title{Haskelite: A Tracing Interpreter Based on a Pattern-Matching Calculus}

\author{Pedro Vasconcelos}
\orcid{0000-0002-8387-9772}
\affiliation{%
  \institution{University of Porto}
  \city{Porto}
  \country{Portugal}
}
\email{pbvascon@fc.up.pt}

\author{Rodrigo Marques}
\orcid{0000-0003-2492-0197}
\affiliation{%
  \institution{University of Porto}
  \city{Porto}
  \country{Portugal}
}
\email{rodrigo.marques@fc.up.pt}
\input{abstract}
\maketitle

\begin{acks}
  This work was financially supported by the Artificial Intelligence
  and Computer Science Laboratory (LIACC) under the base research grant
  \grantnum{FCT}{UIDB/00027/2020}{}
  of the \grantsponsor{FCT}{Funda\c{c}\~{a}o para a Ci\^{e}ncia e
    Tecnologia}{https://www.fct.pt} funded by national funds
  through the FCT/MCTES (PIDDAC).
\end{acks}

\input{intro}

\input{syntax}

\input{bigstep}
\input{smallstep}

\input{soundness}

\input{implementation}

\input{related}
\input{conclusion}

\bibliography{bibliography}
\bibliographystyle{plainnat}

\end{document}

%% file: abstract.tex
\begin{abstract}
  Many Haskell textbooks explain the evaluation of pure functional
  programs as a process of stepwise rewriting using equations.
  However, usual implementation techniques perform program
  transformations that make producing the corresponding tracing
  evaluations difficult.
  
  This paper presents a tracing interpreter for a subset of Haskell
  based on the pattern matching calculus of Kahl. We start from a
  big-step semantics in the style of Launchbury and develop a
  small-step semantics in the style of Sestoft's machines.  This
  machine is used in the implementation of a step-by-step educational
  interpreter. We also discuss some implementation decisions and
  present illustrative examples.
\end{abstract}


%% file: intro.tex
\section{Introduction}\label{sec:intro}
Many commonly used Haskell textbooks explain the evaluation of pure
functional programs as a process of stepwise rewriting using
equations~\cite{bird_wadler_1988, thompson_2011, bird_2015,
  hutton_2016}.  These traces are presented as sequences of
expressions were each step follows from the previous one by a program
equation or some primitive operation.  However, practical
implementations such as GHC perform program transformations that lead
to executions that do not directly relate to the source program.
These transformations make it difficult to automatically generate
execution traces similar to the ones in textbooks.  We believe that
such traces can be helpful while learning functional programming,
particularly in a language with an unfamiliar evaluation strategy, but
there is little tooling to support this available to beginners.  To
fill this gap, we have developed \emph{Haskelite}, a web-based tracing
interpreter for a subset of Haskell. Our goal is to create an
educational tool that complements rather than replaces GHC/GHCi.

For example, consider the canonical definition  of
insertion into an ordered list:
\begin{verbatim}
insert x [] = [x]
insert x (y:ys) | x<=y = x:y:ys
                | otherwise = y:insert x ys
\end{verbatim}
Evaluating an expression such as \verb|insert 3 [1,2,4]| in Haskelite
produces the trace in Figure~\ref{fig:eval-insert}, where each
evaluation step is explained as either a primitive or the use of an
equation.\footnote{The ellipsis (\texttt{...}) represent a pending continuation,
  for example, for evaluating guards.}

\begin{figure}
\begin{Verbatim}[commandchars=\\\{\}]
   insert 3 [1, 2, 4] 
  \remark{\{ 3 <= 1 = False \}}
= .... False
  \remark{\{ insert x (y:ys) | otherwise = y:insert x ys \}}
= 1 : (insert 3 [2, 4])
  \remark{\{ 3 <= 2 = False \}}
= .... False
  \remark{\{ insert x (y:ys) | otherwise = y:insert x ys \}}
= 1 : (2 : (insert 3 [4]))
  \remark{\{ 3 <= 4 = True \}}
= .... True
  \remark{\{ insert x (y:ys) | x<=y = x:y:ys \}}
= 1 : (2 : (3  : (4 : [])))
  \remark{\{ final result \}}
= [1, 2, 3, 4]
\end{Verbatim}
\caption{Evaluation trace for \texttt{insert 3 [1,2,4]}.} \label{fig:eval-insert}
\end{figure}

\begin{figure}
\begin{Verbatim}[commandchars=\\\{\},commentchar=\%]
  head (isort [3, 2, 1])
  \remark{\{ isort = foldr insert [] \}}
= .... foldr insert [] [3, 2, 1]
  \remark{\{ foldr f z (x:xs) = f x (foldr f z xs) \}}
= .... insert 3 (foldr insert [] [2, 1])
  \remark{\{ foldr f z (x:xs) = f x (foldr f z xs) \}}
= ........ insert 2 (foldr insert [] [1])
  \remark{\{ foldr f z (x:xs) = f x (foldr f z xs) \}}
= ............ insert 1 (foldr insert [] [])
  \remark{\{ foldr f z [] = z \}}
= ................ []
  \remark{\{ insert x [] = [x] \}}
= ............ [1]
  \remark{\{ 2 <= 1 = False \}}
= ............ False
  \remark{\{ insert x (y:ys) | otherwise = y:insert x ys \}}
= ........ 1:(insert 2 [])
  \remark{\{ 3 <= 1 = False \}}
= ........ False
  \remark{\{ insert x (y:ys) | otherwise = y:insert x ys \}}
= .... 1:(insert 3 (insert 2 []))
  \remark{\{ head (x:_) = x \}}
= 1
\end{Verbatim}
  \caption{Evaluation trace for \texttt{head (isort [3,2,1])}.}\label{fig:ins-sort}
\end{figure}

Besides equational reasoning, Haskelite can also be used to illustrate
operational aspects of lazy evaluation: consider the list insertion
sort defined as a fold of the previous \texttt{insert} function:
\begin{verbatim}
isort = foldr insert []
\end{verbatim}
Figure~\ref{fig:ins-sort} shows the evaluation trace for the
composition of \texttt{head} with \texttt{isort}.  Because of lazy
evaluation, this computation does \emph{not} require sorting the entire list,
instead performing just $O(n)$ steps (essentially computing the
minimum).  We believe that the ability to easily show such evaluation
traces can be helpful for teaching and learning Haskell.

The implementation of Haskelite required the development of an
alternative operational model, namely, a lazy operational semantics
based on the pattern matching calculus of Kahl~\cite{kahl_2004}.  We
start with a big-step semantics in the style of
Launchbury~\cite{launchbury_1993} and develop a small-step semantics
in the style of Sestoft's machines~\cite{sestof_1997} that forms
the foundation for the Haskelite interpreter.

The contributions presented in this paper are:
\begin{enumerate}
\item the definition of \lambdaPMC, a small core language
  based on a pattern matching calculus;
\item a big-step operational semantics for \lambdaPMC ;
\item an abstract machine derived from the big-step semantics;
\item the description of an implementation of this machine for
  a tracing interpreter of a subset of Haskell.
\end{enumerate}

The remainder of this paper is structured as follows:
Section~\ref{sec:syntax} defines the syntax of \lambdaPMC\ and reviews
the pattern-matching calculus it is based
on. Section~\ref{sec:bigstep} defines normal forms and presents a
big-step semantics for language.  Section~\ref{sec:smallstep} defines
configurations and small-step reduction rules of an abstract machine
for \lambdaPMC\ and discusses the correspondence between the two
semantics.  Section~\ref{sec:implementation} describes the
implementation.  Section~\ref{sec:related} discusses related
work. Finally, Section~\ref{sec:conclusion} highlights directions for
further work.


%% file: syntax.tex
\section{Syntax}\label{sec:syntax}

\subsection{Motivation}\label{sec:patterns}
Handling pattern matching with nested or overlapping patterns in a
strict language is straightforward because the rules for matching can
be separated from the rest of evaluation (see, for example, Chapter 8
of~\cite{ramsey_2022}).  However, in a lazy language, \emph{pattern
  matching forces evaluation} and hence matching and evaluation must
be intertwined.

Established implementation techniques for pattern matching in lazy
languages solve this problem by translating nested patterns into
\emph{simple cases}~\cite{augustsson_1985,spj_1987,jones_1992}, i.e.\@
case expressions with patterns consisting of a single constructor
applied to variables.  Evaluating simple cases requires the evaluation
of the scrutinized expression only to the outermost constructor
(i.e.\@ weak head normal form).  Consider, for example, the following
Haskell function, which checks whether a list has fewer than two
elements:
\begin{verbatim}
isShort (x:y:ys) = False
isShort ys       = True
\end{verbatim}
The translation into simple cases is:
\begin{verbatim}
isShort xs = case xs of
                (x:xs') -> case xs' of
                              (y:ys) -> False
                              [] -> True
                [] -> True
\end{verbatim}
Multiple equations and nested patterns such as \verb|(x:y:ys)| must be
translated into nested case expressions with simple patterns and
matches are completed by introducing missing constructors.

By contrast, the translation of \textit{isShort} into \lambdaPMC\
is direct:
\[ 
  \textit{isShort} = \lambda (\matchpat{(x:y:ys)}{\matchreturn{\textsf{False}}} \mid
  \matchpat{ys}{\matchreturn{\textsf{True}}}) 
\]
Compared to the version with simple cases, the \lambdaPMC\ version
preserves a closer relation to the original source program: each
equation corresponds to one alternative in a \emph{matching
  abstraction} and nested patterns are preserved.  The downside is
that evaluation becomes more complex, because it may have to force
sub-expressions while matching and keep track of alternatives to try in case
of failure.  We claim that the trade-off is beneficial in an 
interpreter where performance is not critical.

\begin{figure}
  \[
    \begin{array}{rcll}
      e & &     \text{expressions} \\
      e & ::= &  x  & \text{variable} \\
      & | &  \app{e_1}{e_2} & \text{function application} \\
      & | & \abstr{m} & \text{matching abstraction} \\
      & | &  \constr{c}{e_1,\ldots,e_n} & \text{constructor application} \\[2ex]
      m &  & \text{matchings} \\
      m & ::= & \matchreturn{e} & \text{return expression}\\
      & | & \matchfail & \text{failure} \\
      & | & \matchpat{p}{m} & \text{match pattern} \\
      & | & \matcharg{e}{m} & \text{argument supply}\\
        & | & \matchalt{m_1}{m_2} & \text{alternative}\\
      &|& \matchwhere{m}{binds} & \text{local bindings} \\[2ex]
      binds && \text{bindings} \\
      binds &::=& \{ x_1=e_1; \ldots; x_n = e_n\} \\[2ex]
      p & & \text{patterns}\\
      p & ::= & x & \text{variable pattern} \\
      & | & \constr{c}{p_1,\ldots,p_n} & \text{constructor pattern} 
  \end{array}
\]
\caption{Syntax of \lambdaPMC.}\label{fig:syntax}
\end{figure}

\subsection{Expressions and Matchings}
Figure~\ref{fig:syntax} defines \lambdaPMC, a small functional
language with syntactical categories for \emph{expressions},
\emph{matchings} and \emph{patterns}. Matchings and patterns are based
on the PMC calculus of Kahl~\cite{kahl_2004}.  The principal
syntactical extension is a \awhere\ construct for defining (possibly)
recursive local bindings.

A matching can be a return expression $\matchreturn{e}$ (signaling a
successful match), the matching failure \matchfail, a pattern match
$\matchpat{p}{m}$ (expecting a single argument that must match pattern
$p$), an argument application $\matcharg{e}{m}$, or the alternative
\matchalt{m_1}{m_2} between two matchings.  Note that patterns $p$ can
be nested; for example \verb|(x:(y:ys))| is a valid pattern (using
infix notation for the list constructor).

Matching abstraction subsumes lambda abstraction:
$\lambda x.\, e$ is equivalent to
$\lambda (\matchpat{x}{\matchreturn{e}})$.

Like Launchbury's semantics, we introduce a syntactical construct to
define local bindings; this is needed to implement lazy evaluation by
guaranteeing that evaluation results are shared (see
Section~\ref{sec:normalized}).  Unlike Lauchbury's semantics, local
bindings are introduced in \emph{matchings} rather than expressions.
By analogy with the Haskell language, we use ``\awhere'' instead of
``\alet'' because the scope of the bindings scope extends beyond a
single expression.

The expression-level \alet\ can be obtained by a straightforward
translation into a \awhere:
 \begin{equation}
   \llet{binds}{e} ~\equiv~
   \lambda(\matchwhere{\matchreturn{e}}{binds})
 \end{equation}
However, due to the potential for matching failure, the reverse translation is
 \emph{not} direct. In other words, where-bindings cannot be readily
 translated into a \alet\ (see the final example in Section~\ref{sec:examples}).
 
Case expressions can also be obtained by
translation into matching abstractions:
\begin{equation}
  \begin{split}
  \textsf{case}~ e_0 ~\textsf{of}~\{p_1\to e_1;\ldots; p_n\to e_n\} \\
 ~\equiv~ 
 \lambda(\matcharg{e_0}{(\matchpat{p_1}{\matchreturn{e_1}}\amatchalt\ldots \amatchalt\matchpat{p_n}{\matchreturn{e_n}})})
 \end{split}
\end{equation}
Note that the above translation is valid even if the patterns
$p_1,\ldots\, p_n$ are nested, overlapping or non-exhaustive.
Also note that we can omit parenthesis because $\amatchalt$ is associative.

As in the STG machine~\cite{jones_1992}, we assume that constructors
are always \emph{saturated}; partial application of constructors can be expressed using a
matching abstraction. For example, the curried list constructor is
\begin{equation}
 (\texttt{:}) ~\equiv~ \lambda(\matchpat{x}{\matchpat{y}\matchreturn{\texttt{:}(x,y)}})
\end{equation}

For readability, we will write list constructors using infix notation,
e.g.\@ use $x\texttt{:}y$ instead of $\texttt{:}(x,y)$.  We will also
present some examples using integers and arithmetic operations. These
are straightforward to define and therefore omitted from the syntax
of Figure~\ref{fig:syntax} and the
semantics of Sections~\ref{sec:bigstep} and~\ref{sec:smallstep} to avoid
unnecessary complexity.

\subsection{Reduction Relations}\label{sec:reduction}
The pattern matching calculus is defined as two \emph{redex}
reduction relations, namely, $\expred$ between expressions and
$\matred$ between matchings. Kahl describes the calculus and proved
its confluence irrespective of reduction
strategies~\cite{kahl_2004}. In Section~\ref{sec:bigstep} we will
define a call-by-need (i.e.\@ lazy) reduction strategy, but for now we
will just review the reduction rules that apply to our setting.

We use the notation $m[e'/x]$ for the substitution of free occurrences
of variable $x$ for an expression $e'$ in the matching $m$.
The definitions of free occurrences and substitution are standard and
therefore not included here.

The first two rules define failure as the left unit of alternative,
and return as the left zero:
\begin{align}
  \matchalt{\matchfail}{m} &~\matred~ m \tag{$\matchfail\amatchalt$}  \\
  \matchalt{\matchreturn{e}}{m} &~\matred~ \matchreturn{e} \tag{$\matchreturn{}\amatchalt$}
\end{align}

The next rule states that matching an abstraction built from a return
expression reduces to the underlying expression.
\begin{equation}
  \lambda \matchreturn{e} ~\expred~ e  \tag{$\lambda\matchreturn{}$} 
\end{equation}

There is no reduction rule for $\lambda \matchfail$: this is because
matching failure corresponds to a runtime error that cannot be caught
in \lambdaPMC.\footnote{Kahl considers a variant of his pattern
  calculus where $\lambda \matchfail$ reduces to an ``empty
  expression'' that can be converted back to matching failure.  We do
  not require this for our semantics.}
  
Application of a matching abstraction reduces to argument supply to
the matching.  Dually, argument supply to a return expression reduces
to application of the
expression.\footnote{Rule $(\amatcharg\matchreturn{})$ is reasonable
  in a higher-order language because the returned expression can be a
  function.}
\begin{align}
  (\lambda m)~ a &~\expred~ \lambda(\matcharg{a}{m}) \tag{$\lambda @$}\\
  \matcharg{a}{\matchreturn{e}} &~\matred~ \matchreturn{e~a} \tag{$\amatcharg\matchreturn{}$} 
\end{align}

The following rule propagates a matching failure through an argument supply.
\begin{equation}
  \matcharg{e}{\matchfail} ~\matred~ \matchfail \tag{$\amatcharg\matchfail$}
\end{equation}

Next, argument supply distributes through alternatives.
\begin{equation}
  \matcharg{e}{(\matchalt{m_1}{m_2})} ~\matred~
  \matchalt{(\matcharg{e}{m_1})}{(\matcharg{e}{m_2})} \tag{$\amatcharg\amatchalt$}
\end{equation}
Note that the above rule duplicates the expression $e$; in the lazy
semantics of Section~\ref{sec:bigstep} we will restrict $e$ to be a
variable so that the result of evaluation can be shared
(thus preventing duplication of work).
For now, however, we are concerned only with the
pure calculus regardless of a particular reduction strategy.

The following rules handle argument supply to patterns; the successful
match ($c\amatcharg c$) decomposes the arguments into the nested
patterns, while the non-successful match ($c\amatcharg c'$) fails
immediately. Finally, matching a variable succeeds and performs substitution.
\begin{gather*}
 \begin{split} \matcharg{\constr{c}{e_1,\ldots,e_n}}{\matchpat{\constr{c}{p_1,\ldots,p_n}}{m}}      
   ~\matred \\ e_1 \amatcharg p_1 \amatchpat \ldots \amatchpat e_n \amatcharg p_n \amatchpat m
 \end{split} \tag {$c \amatcharg c$}\\[1ex]
 \begin{split}
  \matcharg{\constr{c'}{e_1,\ldots,e_k}}{\matchpat{\constr{c}{p_1,\ldots,p_n}}{m}} 
  ~ \matred ~ \matchfail \\ \text{if}~c\neq c'
  \end{split}
\tag{$c \amatcharg c'$} \\[1ex]
 \matcharg{e}{\matchpat{x}{m}} ~\matred~ m[e/x] \tag{$\amatcharg x$}
\end{gather*}

\emph{A clarification on notation}: because of the syntax of
Figure~\ref{fig:syntax}, a matching such as
$$ e_1 \amatcharg p_1 \amatchpat \ldots \amatchpat e_n \amatcharg p_n 
\amatchpat m $$
can only be interpreted as associating to the right:
$$ e_1 \amatcharg (p_1 \amatchpat (e_2 \amatchpat (p_2 \ldots \amatchpat e_n \amatcharg (p_n
\amatchpat m)\ldots)) $$
We therefore omit the redundant parentheses in the rule above and in the
remaining presentation.

\subsection{Translating Haskell into \lambdaPMC}\label{sec:examples}
We will now show some examples of translating a subset of Haskell
definitions into \lambdaPMC.  The emphasis is on developing intuitions
rather than presenting formal rules (this would also 
require a formal definition of Haskell's syntax which is beyond the
scope of this paper). The translation is automatic and can be
done during parsing (see Section~\ref{sec:implementation}).

\paragraph{Overlapping patterns}
Consider the following definition of the \emph{zipWith} function that
combines two lists using a functional argument:
\begin{verbatim}
zipWith f (x:xs) (y:ys) = f x y : zipWith f xs ys
zipWith f xs     ys     = []
\end{verbatim}
This definition  translates directly to \lambdaPMC:
\[
  \begin{array}{l}
  zipWith = \\
  \quad\lambda \!\!\begin{array}[t]{l}
        (f \amatchpat (x:xs) \amatchpat (y:ys) \amatchpat 
        \matchreturn{f~x~y : zipWith~ f~ xs~ ys}  \\
    \amatchalt f \amatchpat xs  \amatchpat ys \amatchpat \matchreturn{[\,]})
              \end{array}
              \end{array}
\]
Overlapping patterns are handled directly because matching alternatives are
tried in left-to-right order; see rules $(\matchfail\amatchalt)$ and
$(\matchreturn{}\amatchalt)$ of Section~\ref{sec:reduction}.

\paragraph{As-patterns and Boolean guards}
Next consider a function \emph{nodups}, which removes identical
consecutive elements from a list, expressed using an
\emph{as-pattern} and a \emph{Boolean guard}:
\begin{verbatim}
nodups (x:xs@(y:ys)) | x==y = nodups xs
nodups (x:xs) = x:nodups xs
nodups [] = []
\end{verbatim}
This can be translated to \lambdaPMC\
as a matching abstraction with nested alternatives:
\[
  \begin{array}{l}
    nodups = 
    \lambda
  \!\!\begin{array}[t]{l}
    ((x:xs) \amatchpat
        xs \amatcharg (y:ys) \amatchpat \\
         \qquad (x==y) \amatcharg \textsf{True} \amatchpat   \matchreturn{nodups~ xs}  \\
    \amatchalt (x:xs) \amatchpat  \matchreturn{x: nodups~xs} \\
    \amatchalt [\,] \amatchpat \matchreturn{[\,]})
      \end{array}
    \end{array}
\]
Following Kahl~\cite{kahl_2004}, the Boolean guard is encoded as the
evaluation of the expression $(x==y)$ matched against the constructor
\textsf{True}.

Note that \lambdaPMC\ could also be used to encode more advanced
pattern matching extensions, such as \emph{pattern guards}, 
\emph{view patterns} and \emph{lambda cases}. However,
we do not require this for our interpreter, so we omit the details
from this presentation.

\paragraph{Local bindings}
Our final example illustrates the use of local bindings
in \lambdaPMC. Consider the following (contrived) Haskell function: 
\begin{verbatim}
foo x y 
    | z>0 = z+1
    | z<0 = z-1
    where z = x*y
foo x y   = x+y
\end{verbatim}
Note that the binding for $z$ scopes over both Boolean guards and
corresponding right-hand sides but not over the last equation. Also,
note that the guards are not exhaustive: the last equation applies
when $z=0$.

We can translate this into \lambdaPMC\ preserving the scoping and the
matching semantics using a \awhere\ binding:
\[ \textit{foo} = \lambda
  \!\!\begin{array}[t]{l}
     (x \amatchpat y \amatchpat
     \!\!\begin{array}[t]{l}
       ( (z > 0 \amatcharg \textsf{True} \amatchpat \matchreturn{z+1})\\
       \amatchalt (z < 0 \amatcharg \textsf{True} \amatchpat \matchreturn{z-1}))\\
       \awhere\ z=x*y
     \end{array}\\
     \amatchalt x \amatchpat y\amatchpat \matchreturn{x+y})
   \end{array}\]
 
 At first glance, it may appear that one could have also expressed
 this using \alet:
\[ \textit{foo}' = \lambda
  \!\!\begin{array}[t]{l}
        (x \amatchpat y \amatchpat
        \lreturn \alet~ z = x*y ~ \\
        \qquad\qquad\quad\ain~\lambda
        \!\!\begin{array}[t]{l}
          (z>0 \amatcharg \textsf{True} \amatchpat
          \matchreturn{z+1})\\
          \amatchalt z<0 \amatcharg \textsf{True} \amatchpat \matchreturn{z-1})\rreturn
        \end{array} \\
     \amatchalt x \amatchpat y\amatchpat \matchreturn{x+y})
      \end{array}\]
However, $\textit{foo}'$ is not equivalent to $\textit{foo}$: when $z=0$ then
the inner lambda matching fails and the second equation is not
tried (meaning $\textit{foo}'$ is undefined for $x, y$ such that $x*y=0$).
This is because a return expression such as
$\matchreturn{\llet{z=x*y}{\ldots}}$ is a committed
choice, absorbing any remaining alternatives.\footnote{See rule
  $(\matchreturn{}\amatchalt)$ in Section~\ref{sec:reduction}.}

%% file: bigstep.tex
\section{Big-step Semantics}\label{sec:bigstep}

In this section, we define a big-step semantics for
\lambdaPMC.  The semantics is based in Sestoft's
revision~\cite{sestof_1997} of Launchbury's semantics for lazy
evaluation~\cite{launchbury_1993}. The definitions of
\emph{matching arity} and \emph{weak head normal forms} in Section~\ref{sec:whnf}
are novel, as are the rules for matchings.

\subsection{Normalized Syntax}\label{sec:normalized}
As in Launchbury's semantics, we will restrict arguments of
applications to be simple variables rather than arbitrary expressions;
this is needed to implement call-by-need by updating a heap with
results of evaluations.  We need to restrict arguments of applications
not just in expressions but also in matchings\footnote{Rule
  $(\amatcharg\amatchalt)$ of Section~\ref{sec:reduction} distributes
  a argument matching over alternatives; by restricting the argument to be a
  variable we ensure that any evaluation result is shared.}:
\[
    \begin{array}{rcl}
      e & ::= & \cdots\quad  |\quad  \app{e}{y} \quad|\quad
                \constr{c}{y_1,\ldots,y_n} \quad|\quad
                \cdots  \\[1ex]
      m & ::= & \cdots\quad |\quad \matcharg{y}{m} \quad|\quad \cdots 
    \end{array}
  \]
It is always possible to translate arbitrary \lambdaPMC\ terms into
these forms by using \alet\ and \awhere\ bindings:
\begin{align*}
  \constr{c}{e_1,\ldots,e_n} &~\leadsto~
                               \llet{\{y_i = e_i\}}{\constr{c}{y_1,\ldots,y_n}}\\
  \app{e_1}{e_2} & ~\leadsto~
                   \llet{y=e_2}{(\app{e_1}{y})} \\
  \matcharg{e}{m} & ~\leadsto~ \matchwhere{(\matcharg{y}{m})}{y=e}
\end{align*}
where the variables $y_i$ are suitably ``fresh''.

Note that normalized syntax still allows nested constructors in patterns;
indeed, the whole purpose of the semantics presented here is to treat
such matchings directly.

\subsection{Preliminary Definitions}\label{sec:whnf}

Let \emph{heaps} $\Gamma, \Delta, \Theta$ be finite mappings from
variables to (possibly unevaluated) expressions.  The notation
$\Gamma[\loc \mapsto e]$ represents a heap that maps variable $\loc$ to $e$
and otherwise behaves as $\Gamma$.

Because matchings can encode multi-argument functions and also
argument applications, we have to define the results of evaluations
accordingly. First we define a syntactical measure \arity{m} of the \emph{arity}
of a matching:
\[
  \begin{array}{rcl}
    \arity{\matchreturn{e}}         &=& 0 \\
    \arity{\matchfail}              &=& 0\\
    \arity{(\matchpat{p}{m})}       &=& 1 + \arity{m} \\
    \arity{(\matcharg{y}{m})}       &=& \max(0, \arity{m} - 1) \\
    \arity{(\matchalt{m_1}{m_2})}   &=& \arity{m_1}\,, \qquad \text{if}~ \arity{m_1} = \arity{m_2} \\
                                    & & \hfill \text{(otherwise undefined)}\\
    \arity{(\matchwhere{m}{binds})} &=& \arity{m}
  \end{array}
\]
The penultimate equation states that matchings in alternatives must
have equal arities; this generalizes the Haskell syntax rule requiring
that equations for a binding should have the same number of
arguments~\cite{haskell_2010_report}.

An expression is in weak head normal (whnf) $\whnf$ if it is either
a matching abstraction of arity greater than zero or a constructor:
\[
\begin{array}{rcll}
  \whnf  &::=&  \lambda m & \text{such that}~ \arity{m}> 0 \\
     &|& \constr{c}{y_1,\ldots,y_n}
\end{array}
\]
If $\arity{m}>0$ then $m$ expects at least one argument, i.e.\@
behaves like a lambda abstraction.  If $\arity m= 0$ then $m$ is
\emph{saturated} (i.e.\@ fully applied) and therefore \emph{not} a
whnf.  This definition of weak normal forms implies that partially
applied matchings will be left unevaluated e.g.\@
$\lambda(\matcharg{z}{\matchpat{x}{\matchpat{y}{\matchreturn{e}}}})$
is in whnf. This is similar to how partial applications are handled in
abstract machines with multi-argument functions e.g.\@ the
STG~\cite{jones_1992}.  However, in our case, this is not done for
efficiency, but rather to ensure that we only evaluate \emph{saturated}
matchings. The results $\matchresult$ of evaluating such matchings are
either a return expression or a failure:
\[
  \matchresult \quad::=\quad \matchreturn{e} \quad|\quad \matchfail
\]

\begin{figure}
  \begin{gather*}
    %
    %
    \prooftree
    \justifies
     \Gamma; \lset; w \expev \Gamma; w 
    \using{\bigrule{Whnf}}
    \endprooftree\\[3ex]
    \prooftree
     \arity{m}= 0 \qquad
    \Gamma; \lset; [\,];\, m \matev \Delta;\matchreturn{e} \qquad
    \Delta; \lset; e \expev \Theta;\,w 
    \justifies
     \Gamma;\lset; \lambda m \expev \Theta;w 
    \using{\bigrule{Sat}}
    \endprooftree\\[3ex]
     \prooftree
      \Gamma;\lset\cup\{\loc\}; e \expev \Delta; w 
     \justifies
      \Gamma[\loc\mapsto e]; \lset; \loc \expev \Delta[\loc\mapsto w]; w 
     \using{\bigrule{Var}}
     \endprooftree \\[3ex]
     \prooftree
      \Gamma;\lset; e\expev \Delta;\lambda m 
     \qquad
      \Delta;\lset; \lambda (\matcharg{\loc}{m}) \expev \Theta;w 
     \justifies
      \Gamma;\lset; (e~\loc) \expev \Theta; w 
     \using{\bigrule{App}}
    \endprooftree 
     %
     %
  \end{gather*}
  \caption{Expression evaluation}\label{fig:expr-eval}
\end{figure}

\begin{figure}
  \begin{gather*}
    %
    %
    \prooftree
    \justifies
    \Gamma; \lset; \argstack; \matchreturn{e} \matev \Gamma; \matchreturn{\matcharg{\argstack}{e}}
    \using{\bigrule{Return}}
    \endprooftree \\[3ex] 
    \prooftree
    \justifies
    \Gamma; \lset; A; \matchfail \matev \Gamma;\matchfail
    \using{\bigrule{Fail}}
    \endprooftree\\[3ex]
    \prooftree
    \Gamma; \lset; (y:\argstack); m \matev \Delta;\matchresult
    \justifies
    \Gamma; \lset; \argstack; \matcharg{y}{m} \matev \Delta; \matchresult
    \using{\bigrule{Arg}}
    \endprooftree\\[3ex]
    \prooftree
    \Gamma;\lset; \argstack; m[y/x] \matev \Delta; \matchresult
    \justifies
    \Gamma; \lset; (y:\argstack); \matchpat{x}{m} \matev \Delta; \matchresult
    \using{\bigrule{Bind}}
    \endprooftree \\[3ex]
    \prooftree
    \begin{array}{l}
    \Gamma;\lset; y\expev \Delta; \constr{c}{y_1,\ldots,y_n} \\
    \Delta;\lset;\argstack; \matcharg{y_1}{p_1} \amatchpat \ldots \amatchpat \matcharg{y_n}{p_n} \amatchpat m
      \matev \Theta;\matchresult
    \end{array}
    \justifies
    \Gamma; \lset; (y:\argstack); \matchpat{\constr{c}{p_1,\ldots,p_n}}{m}
    \matev \Theta;\matchresult
    \using{\bigrule{Cons1}}
    \endprooftree \\[3ex]
    \prooftree
    \Gamma;\lset; y\expev \Delta; \constr{c'}{y_1,\ldots, y_k} \qquad
    c \neq c' 
    \justifies
    \Gamma; \lset; (y:\argstack); \matchpat{\constr{c}{p_1,\ldots, p_n}}{m}
    \matev \Delta; \matchfail
    \using{\bigrule{Cons2}}
    \endprooftree    \\[3ex]
    \prooftree
    \Gamma;\lset;\argstack; m_1 \matev \Delta;\matchreturn{e}
    \justifies
    \Gamma;\lset;\argstack; (\matchalt{m_1}{m_2}) \matev \Delta;\matchreturn{e}
    \using{\bigrule{Alt1}}
    \endprooftree \\[3ex]
    \prooftree
    \Gamma;\lset;\argstack; m_1 \matev \Delta;\matchfail  \qquad
    \Delta;\lset;\argstack; m_2 \matev \Theta; \matchresult
    \justifies
    \Gamma;\lset;\argstack; (\matchalt{m_1}{m_2}) \matev \Theta;\matchresult
    \using{\bigrule{Alt2}}
    \endprooftree \\[3ex]
    \prooftree
    \Gamma[y_i\mapsto \widehat{e}_i];\lset;\argstack; \widehat{m} \matev
    \Delta; \matchresult
    \justifies
    \Gamma;\lset;\argstack; (\matchwhere{m}{\{x_i=e_i\}}) \matev
    \Delta; \matchresult
    \using{\bigrule{Where}}
    \endprooftree \quad (y_i~\text{are fresh})
  \end{gather*}
  
  \caption{Matching evaluation}\label{fig:match-eval}
\end{figure}

\subsection{Evaluation Rules}

Evaluation is defined in Figures~\ref{fig:expr-eval}
and~\ref{fig:match-eval} by two mutually recursive judgments:
\begin{description}
  \item[$\Gamma;\lset;e \expev \Delta;w$]  Evaluating 
   expression $e$ from heap $\Gamma$ yields heap $\Delta$ and result $w$;
  \item[$\Gamma;\lset;\argstack; m \matev \Delta;\matchresult$] 
    Evaluating matching $m$ from heap $\Gamma$ and argument stack \argstack\
    yields heap $\Delta$ and result \matchresult.
  \end{description}

  In the $\matev$ judgments the argument stack \argstack\ is
  a sequence of variables representing the pending arguments
  to be applied to the matching expression.

  In both judgments the set $\lset$ keeps track of variables under
  evaluation and is used to ensure freshness of variables in the
  \bigrule{Where} rule.
  \begin{definition}[Freshness condition, adapted from~\cite{sestof_1997}]
    A variable $y$ is fresh with respect to
    a judgment 
    $\Gamma;\lset;\argstack;m \matev \Delta;\matchresult$
    if it does not occur (free or bound) in
     $\Gamma,\lset,\argstack$ or $m$.
  \end{definition}

\paragraph{Remarks about rules for  expressions (Figure~\ref{fig:expr-eval}).}
  
Rule \bigrule{Whnf} terminates evaluation immediately when we reach a
weak head normal form, namely a non-saturated matching abstraction or a
constructor.
  
Rule~\bigrule{Sat} applies to saturated matching abstractions; if the
matching evaluation succeeds then we proceed to evaluate the
returned expression.
    
Rule~\bigrule{Var} forces the evaluation of an expression in the heap;
as in Launchbury and Sestoft's semantics, we remove the entry from the
heap during evaluation (``black-holing'') and update the heap with the
result afterwards.
    
Rule~\bigrule{App} first evaluates the function to obtain a matching
abstraction and then evaluates the argument application.
      
  \paragraph{Remarks about rules for matchings (Figure~\ref{fig:match-eval}).}

  Rule \bigrule{Return} terminates evaluation successfully
  when we reach a return expression.
  The notation $A\amatcharg e$ represents the nested applications
  of  left over arguments on the stack $A$ to the expression $e$.
  The definition is as follows:
  \begin{align*}
  \matcharg{[\,]}{e} &= e \\
  \matcharg{(y:ys)}{e}  &= \matcharg{ys}{(e~y)}
  \end{align*}

  Rule \bigrule{Fail} terminates evaluation unsuccessfully when
  we reach the match failure \matchfail.

  Rule~\bigrule{Arg} pushes an argument onto the stack and carries
  on evaluation.
  
  Rule~\bigrule{Bind} binds a variable pattern to
  an argument on the stack. This is simply a renaming of the pattern
  variable $x$ to the heap variable $y$.
  
  Rules~\bigrule{Cons1} and \bigrule{Cons2} handle the successful
  and unsuccessful matching of a constructor pattern; in the later
  case the matching evaluation returns \matchfail.  Note also that
  \bigrule{Cons1} decomposes the matching of sub-patterns
  in left-to-right order (i.e.\@ first matching $p_1$, then $p_2$, etc.).

  Rules~\bigrule{Alt1} and \bigrule{Alt2} handle progress and failure
  in alternative matchings.  Note that in \bigrule{Alt2} evaluation
  continues with the updated heap $\Delta$ because the effects of evaluation
  of the failed match must be preserved when evaluating $m_2$.
  Note also that the argument stack is shared between $m_1$ and $m_2$
  implementing rule $(\amatcharg\mid)$ of Section~\ref{sec:reduction}.
  This sharing justifies why we restrict argument supply to single
  variables: allowing arbitrary expressions as arguments could
  duplicate computations because the same expression could be evaluated
  in more than one alternative. By restricting to single variables
  we can ensure that results are shared.

  Rule~\bigrule{Where} is analogous to the one by Sestoft for
  \alet: it allocates new unevaluated expressions in the heap (taking
  care of renaming) and continues evaluation of the body.
  We use $\widehat{e_i}$ and $\widehat{m}$ for the renaming of
  bound expressions and body, respectively:
  \begin{gather*}
    \widehat{e_i} = e_i [y_1/x_1,\ldots,y_n/x_n] \\
    \widehat{m} = m [y_1/x_1,\ldots,y_n/x_n] 
  \end{gather*}


%% file: smallstep.tex
\section{Abstract Machine}\label{sec:smallstep}

We will now transform the big-step semantics of
Section~\ref{sec:bigstep} into an abstract
machine, i.e.\@ a transition function between configurations.
Each transition performs a bounded amount of work, which suits
the implementation of a step-by-step interpreter.  The
presentation follows the derivation done by
Sestoft~\cite{sestof_1997}.

\subsection{Configurations}
Because our big-step semantics has two mutually recursive judgments,
we introduce a \emph{control} component that keeps
track of the current evaluation mode.
\[ \begin{array}{lcll}
     C & && \text{control} \\
     C &::=& \eval e & \text{evaluate expression}~e\\
     &\mid& \match{\argstack}{m} & \text{evaluate matching}~m~\text{with arguments}~A 
\end{array} \]

The next step is to make evaluation order explicit 
in a \emph{stack}, which is a list of \emph{continuations} $\kont$:
\[ \begin{array}{lcll}
     \kont &&& \text{continuations} \\
     \kont &::=& y & \text{push argument} \\
           &\mid& !y & \text{push update} \\
           &\mid& \$ & \text{end matching} \\
           &\mid& ?(\argstack,m) & \text{push alternative} \\
           &\mid& @(\argstack,\matchpat{\constr{c}{\vec{p}}}{m}) & \text{push pattern}
   \end{array}
   \]

   A machine configuration is a triple $(\Gamma,C,\retstack)$ of heap,
   control, and return stack.  The initial configuration for
   evaluating $e$ is $(\{\},\, \eval e,\, [\,])$. Evaluation may
   terminate successfully in a configuration
   $(\Gamma, \eval w, [\,])$, get ``stuck'' in a configuration
   $(\Gamma,\, \match{\argstack}{\matchfail},\, \$:\retstack)$ due to
   pattern matching failure, or continue indefinitely (i.e.\@ diverge).

\begin{figure}
  \[
    \begin{array}{rllll} \hline
      & \text{Heap} & \text{Control} & \text{RetStack} & \text{rule} \\ \hline
      & \Gamma & \eval (e~\loc) & \retstack  & \smallrule{App1} \\
      \Longrightarrow & \Gamma & \eval e  & \loc:\retstack    \\[2ex]
      & \Gamma & \eval \lambda m & \loc:\retstack & \smallrule{App2} \\
      \Longrightarrow & \Gamma & \eval \lambda (\matcharg{\loc}{m}) & \retstack\\
      \multicolumn{5}{l}{\text{if}~\arity{m}>0} \\[2ex]
      & \Gamma & \eval \lambda m & \retstack & \smallrule{Sat} \\
      \Longrightarrow & \Gamma & \match{[\,]}{m} & \$:\retstack  \\
      \multicolumn{5}{l}{\text{if}~ \arity{m}=0}\\[2ex]
      & \Gamma[\loc\mapsto e] & \eval \loc & \retstack & \smallrule{Var} \\
      \Longrightarrow & \Gamma & \eval e  & !\loc : \retstack &   \\[2ex]
      & \Gamma  & \eval w  & !\loc : \retstack  & \smallrule{Update}  \\
      \Longrightarrow & \Gamma[\loc\mapsto w] & \eval w   & \retstack & 
    \end{array}
\]
     
  \caption{Abstract machine rules for expressions.}\label{fig:smallstep-1}
\end{figure}

\begin{figure*}
  \[
    \begin{array}{rllll} \hline
      & \text{Heap} & \text{Control} & \text{RetStack} & \text{Rule} \\ \hline

       & \Gamma & \match {\argstack}{\matchreturn{e}} & \retstack &  \smallrule{Return1A} \\
       \Longrightarrow & \Gamma & \match{[\,]}{\matchreturn{\matcharg{A}{e}}} & \retstack & \\
       \multicolumn{5}{l}{\text{if}~\argstack \neq [\,]} \\[2ex]
       & \Gamma & \match {[\,]}{\matchreturn{e}} & \$:\retstack &  \smallrule{Return1B} \\
       \Longrightarrow & \Gamma & \eval{e} & \retstack & \\[2ex]
       & \Gamma & \match {[\,]}{\matchreturn{e}} & (?(\argstack',m)):\retstack & \smallrule{Return2} \\
       \Longrightarrow & \Gamma & \match {[\,]}{\matchreturn{e}} & \retstack\\[2ex]
       & \Gamma & \match {(y:\argstack)} {(\matchpat{x}{m})}  & \retstack & \smallrule{Bind} \\
       \Longrightarrow & \Gamma & \match{\argstack}{m[y/x]}  & \retstack & \\[2ex]
       & \Gamma & \match  {(y:\argstack)}{(\matchpat{\constr{c}{\vec{p}}}{m})}  & \retstack & \smallrule{Cons1} \\
       \Longrightarrow & \Gamma & \eval y  & @(\argstack,\matchpat{\constr{c}{\vec{p}}}{m}):\retstack & \\[2ex]
       & \Gamma & \eval \constr{c}{y_1,\ldots,y_n}  & @(\argstack,\matchpat{\constr{c}{p_1,\ldots,p_n}}{m}):\retstack &  \smallrule{Cons2} \\
       \Longrightarrow & \Gamma & \match{\argstack}{(\matcharg{y_1}{p_1} \amatchpat \ldots \matcharg{y_n}{p_n} \amatchpat m)}  & \retstack \\[2ex]
       & \Gamma & \eval \constr{c'}{y_1,\ldots,y_k} & @(\argstack,\matchpat{\constr{c}{p_1,\ldots,p_n}}{m}):\retstack &  \smallrule{Fail} \\
       \Longrightarrow & \Gamma & \match {[\,]}{\matchfail}  &  \retstack \\
       \multicolumn{5}{l}{\text{if}~ c\neq c'} \\[2ex]
       & \Gamma & \match{\argstack}{(\matcharg{y}{m})} & \retstack & \smallrule{Arg} \\
       \Longrightarrow & \Gamma & \match{(y:\argstack)}{m}  & \retstack \\[2ex]
       & \Gamma & \match {\argstack}{(\matchalt{m_1}{m_2})} & \retstack & \smallrule{Alt1} \\
       \Longrightarrow & \Gamma & \match{\argstack}{m_1}   &?(\argstack,m_2):\retstack & \\[2ex]
   & \Gamma & \match{\argstack'}{\matchfail}  & ?(\argstack,m):\retstack & \smallrule{Alt2} \\
      \Longrightarrow & \Gamma & \match{\argstack}{m}  & \retstack  \\[2ex]
            & \Gamma & \match{A}{(\matchwhere{m}{\{x_i=e_i\}})} & \retstack & \smallrule{Where} \\
      \Longrightarrow & \Gamma[\loc_i\mapsto \widehat{e_i}] & \match{A}{\widehat{m}}   & \retstack \\
      \multicolumn{5}{l}{\text{where}~\loc_i~\text{are fresh and}~ 
      \widehat{e}_i = e_i[\loc_1/x_1,\ldots,\loc_n/x_n]~\text{and}~
      \widehat{m} = m[\loc_1/x_1,\ldots,\loc_n/x_n]}
    \end{array}
\]
  \caption{Abstract machine rules for matchings.}\label{fig:smallstep-2}
\end{figure*}

\subsection{Transitions}

The transitions between configurations are given by rules in
Figures~\ref{fig:smallstep-1} and~\ref{fig:smallstep-2}.
Meta-variables that appear on the left-hand side and right-hand side
are unmodified. For example: rule~\smallrule{App1} applies on any
heap and stack, and does not modify the
heap, but pushes a single continuation onto the stack.

Rules~\smallrule{App1}, \smallrule{Var}, \smallrule{Update}
are identical to the ones in the first version of
Sestoft's abstract machine~\cite{sestof_1997}.

Rule~\smallrule{App2} and \smallrule{Sat} handle application
and matching evaluations, respectively. Note that the
side conditions on $\arity{m}$ ensure at most one
rule applies.
As in the big-step semantics, rule \smallrule{Sat} switches
from evaluating an expression to a
matching, pushing a mark `\$' onto the return stack
to allow checking when no pending alternatives
are available (rule~\smallrule{Return1B}).

Rule~\smallrule{Cons1} switches from evaluating a matching to an
expression in order to perform a pattern match, pushing a continuation
onto the return stack.  Rule~\smallrule{Cons2} and \smallrule{Fail}
handle the successful and unsuccessful matches.
Rules~\smallrule{Alt1} and \smallrule{Alt2} deal with
alternatives. Rule~\smallrule{Arg} pushes arguments on the local
argument stack.  Finally, rule \smallrule{Where} is analogous to
Sestoft's rule for \alet.


%% file: soundness.tex
\subsection{Consistency Between the Two Semantics}\label{sec:soundness}

Now that we have introduced both a big-step semantics and an abstract
machine, we show they are consistent with each other.
The methodology follows that of Sestoft~\cite{sestof_1997} and uses
induction on the derivations. We present only the definitions and 
main statements; the proofs are  straightforward but tedious.

The first result states that a big-step evaluation corresponds to a
sequence of small-step transitions in the machine.  Because the
evaluation of expressions and matchings are mutually recursive, we must
prove the result for both evaluations simultaneously.  We use
$\Rightarrow^{*}$ for the reflexive and transitive closure of the
small-step transition relation.

\begin{theorem}[$\Downarrow$ implies $\Rightarrow^*$]\label{thm:soundness1}
  If  $\Gamma; \lset; e \expev \Delta; w$
  then, for all $S$,
  $(\Gamma, \eval{e},  \retstack) \Rightarrow^{*}
    (\Delta, \eval{w}, \retstack)$.

  If $\Gamma; \lset; \argstack; m \matev \Delta; \matchresult$
  then, for all $S$,
  $(\Gamma, \match{\argstack}{m}, \retstack) \Rightarrow^{*}
    (\Delta, \match{[\,]}{\matchresult}, \retstack)$
  \end{theorem}
  \begin{proof}
    This result can be proved by simultaneous induction on the height
    of the derivations of \expev\ and \matev.
\end{proof}

The second result establishes that the small-step semantics
derives only evaluations corresponding to the
big-step semantics, provided we restrict ourselves
to \emph{balanced} evaluations.

\begin{definition}
  A sequence of evaluation steps
  $(\Gamma,C,\retstack) \Rightarrow^{*}
  (\Delta,C',\retstack)$ is \emph{balanced} if
  the initial and final stacks are identical and every intermediate
  stack is of the form $S' = \kont_1:\kont_2:\ldots:\kont_n:S$, i.e.\@
  an extension of the initial stack.
\end{definition}

\begin{definition}
  The \emph{trace} of a sequence of small-step evaluation
  steps $(\Gamma,C,\retstack) \Rightarrow^{*}
  (\Delta,C',\retstack')$ is the sequence of rules used in the evaluation.
\end{definition}

\begin{definition}
  Let $\upd{\retstack}$
  be the set of locations marked
  for updates in a stack \retstack, i.e.\@
  $\upd{\retstack} = \{ y ~:~ (!y) \in \retstack \}$.
\end{definition}

\begin{theorem}[$\Rightarrow^*$ implies $\Downarrow$] \label{thm:soundness2}
  If $(\Gamma, \eval{e}, \retstack) \Rightarrow^{*} (\Delta, \eval{\whnf},
  \retstack)$ is a balanced expression evaluation then
  $\Gamma;\upd{\retstack}; e \expev \Delta; w$.

  If $(\Gamma, \match{A}{m}, \retstack) \Rightarrow^{*} (\Delta,
  \match{[\,]}{\matchresult}, \retstack)$ is a balanced matching evaluation then
  $\Gamma;\upd{\retstack}; A;m \matev \Delta;\matchresult$
\end{theorem}

\begin{proof}
The proof is by induction on the height of balanced derivations.
\end{proof}

As a direct consequence of Theorem~\ref{thm:soundness2}, we conclude
that an abstract machine evaluation starting with an empty stack must
correspond to a successful big-step evaluation, with identical
resulting whnf and heap.

\begin{corollary}
  If $(\Gamma, \eval{e}, [\,]) \Rightarrow^* (\Delta, \eval{\whnf}, [\,])$
  then $\Gamma; \emptyset; e \expev \Delta; \whnf$.
\end{corollary}

\begin{proof} The result follows from the first statement of
  Theorem~\ref{thm:soundness2} when $S=[\,]$.
\end{proof}


%% file: implementation.tex
\section{Implementation}\label{sec:implementation}

\subsection{Overview}
\begin{figure*}
\begin{verbatim}
insert x [] = [x]
insert x (y:ys) | x<=y = x:y:ys
                | otherwise = y:insert x ys
\end{verbatim}

  \[ insert =
    \lambda(\!\!
    \begin{array}[t]{l}
     \matchpat{x}{\matchpat{[\,]}{\matchreturn{x:[\,], \remark{"insert x [] = [x]"}}}} \\
      \amatchalt \matchpat{x}{\matchpat{(y:ys)}{\matcharg{x\leq y}{\matchpat{\textsf{True}}{\matchreturn{x:y:ys, \remark{"insert x (y:ys) | x<=y = x:y:ys"}}}}}} \\
      \amatchalt \matchpat{x}{\matchpat{(y:ys)}{\matchreturn{y:insert~x~ys, \remark{"insert x (y:ys) | otherwise = y:insert x ys"}}}})
      \end{array}
    \]
\caption{Annotated translation of insert into \lambdaPMC.}\label{fig:annotated-insert}
\end{figure*}
The abstract machine derived in Section~\ref{sec:smallstep} forms the
basis for \emph{Haskelite}, a web-based interpreter for a small subset
of Haskell.  The interpreter is written in \emph{Elm}\footnote{\url{https://elm-lang.org}} and consists of:
\begin{enumerate}
\item a parser for translating a subset of Haskell into \lambdaPMC\
  (written using Elm's parser combinator library); 
\item a simple type checker based on the Hindley-Milner system;
\item an abstract machine based on the semantics of Section~\ref{sec:smallstep};
\item pretty-printing code for expressions and the evaluation state;
\item a subset of the Haskell Prelude,
  implemented in the Haskelite language itself and bundled with the interpreter.
\end{enumerate}
The interpreter compiles to JavaScript and runs from an
HTML page entirely on the client side (i.e.\@ the web browser).  We
also use an open-source JavaScript editor supporting syntax
highlighting.\footnote{\url{https://codemirror.net/}} The application
requires little resources by today's web standards and runs on
computers, tablets, or even smartphones. The minimized JavaScript file
for the interpreter (excluding the text editor) is around 100~KB.  On
a selection of typical functional programming introductory examples
the web page consumes between 5--30~MB of memory when
running.\footnote{Measured using Firefox's web developer's tools on an
  Ubuntu Linux AMD64 PC.}

The evaluator follows the rules of Figures~\ref{fig:smallstep-1}
and~\ref{fig:smallstep-2} using purely functional data structures for
stacks and heap.  Since evaluations are likely to be of small
examples, there is no garbage collector.\footnote{Also,
  because the interpreter keeps the evaluation history, a garbage
  collector would probably reclaim less free space than a real
  implementation.}  We have also added primitive values and operations
over integers and characters; these straightforward to implement, so we
omit the details in this paper.

The normalization of expressions and matchings
(Section~\ref{sec:normalized}) is done by introducing indirection
bindings for arguments that are not atomic.
Indirection bindings are short-circuited during pretty-printing, so
that the students are unaware of the normalization; this allows
presenting evaluation as if it were performed on an expression rather
than a graph.\footnote{To ensure termination for
  cyclic structures, we limit the short-circuiting to the first
  traversal of a cycle.}

\subsection{Producing Traces}
We avoid showing all transition steps of the abstract machine for two reasons:
\begin{itemize}
\item the transition steps are too fine-grained, resulting in many
  uninteresting intermediate configurations, e.g.\@
  after a partial application (rules \smallrule{App1} and \smallrule{App2});
\item we do not want to expose the details of \lambdaPMC\ to
  students, in particular, the evaluations of matchings which
  cannot be easily translated back into Haskell.
\end{itemize}
Our solution is to show only configurations after two
kinds of transitions:
\begin{enumerate}
\item the evaluation of primitive operations (e.g.\@ integers); and
\item matching evaluations that return an expression
  (rule \smallrule{Return1B}).
\end{enumerate}
Obtaining equations for primitive operations is trivial (e.g.\@
\texttt{2 * 3 = 6}).  For matching evaluations, we simply collect the
equations from the Haskell source during parsing and annotate the
\lambdaPMC\ abstract syntax.  Figure~\ref{fig:annotated-insert} shows
an example translation of \texttt{insert} into annotated \lambdaPMC.

To show configurations we transform the return stack into an
evaluation context around the current expression.  When the return
stack contains pending matchings (i.e.\@ during the evaluation of
guards), we simply hide the remaining continuations and show only the
stack depth as a sequence of dots (see Figures~\ref{fig:eval-insert}
and~\ref{fig:ins-sort}).  This simple solution highlights the current
evaluation context while also limiting the size of the visualized
expressions.

\subsection{Evaluating Constructors to Normal Form}
The semantics of Sections~\ref{sec:bigstep} and~\ref{sec:smallstep}
evaluate only to weak head normal form, i.e.\@ the outermost data
constructor.  In a read-eval-print loop such as GHCi, full evaluation
is done by converting the result into a string.

In Haskelite we implement full evaluation as a function \texttt{force
  :: a -> a} that finds the next outermost redex in a constructor and
continues evaluation. Unlike the Haskell function of the same name,
\texttt{force} is built-in rather than defined in a type class (which
we do not support anyway).

Forcing is implicitly done for the outermost expression under
evaluation; this works even with infinite data structures, e.g.\@
lists, showing evaluation step-by-step.
Furthermore, \texttt{force} can also be explicitly called to
force full evaluation of (finite) intermediate data structures for pedagogical
reasons.

\subsection{Bang Patterns}
We have added GHC's \emph{bang patterns}
extension~\cite{ghc_guide_bang_patterns} to our implementation:
matching a pattern of the form $!x$ succeeds binding a variable $x$
with any value, provided we can reduce it to whnf before continuing.

The changes to the abstract machine consist of just two new rules:
\[ \begin{array}{rlll}
        \Gamma & \match {(y:\argstack)} {(\matchpat{!x}{m})}  & \retstack & \smallrule{Bang1} \\
       \Longrightarrow~ \Gamma & \eval{y}  & !(\argstack,m[y/x]):\retstack & \\[2ex]
        \Gamma & \eval{w}  & !(\argstack,m):\retstack & \smallrule{Bang2} \\
     \Longrightarrow~ \Gamma & \match{\argstack}{m} & \retstack
   \end{array} \]
   
The continuation $!(\argstack,m[y/x])$ records what to do after 
evaluation of $y$.  We handle the binding of $x$ to
$y$ in rule \smallrule{Bang1}; rule~\smallrule{Bang2} applies after
the evaluation of $y$ to whnf (and its result is updated in the heap);
hence, we can simply ignore the whnf $w$ and evaluate the matching
$m$.

\begin{figure*}
\begin{minipage}{0.5\textwidth}
\begin{Verbatim}[commandchars=\\\{\},commentchar=\%]
  foldl (*) 1 [2, 3, 4]
  \remark{\{ foldl f z (x:xs) = foldl f (f z x) xs \}}
= foldl (*) (1 * 2) [3, 4]
  \remark{\{ foldl f z (x:xs) = foldl f (f z x) xs \}}
= foldl (*) ((1 * 2) * 3) [4]
  \remark{\{ foldl f z (x:xs) = foldl f (f z x) xs \}}
= foldl (*) (((1 * 2) * 3) * 4) []
  \remark{\{ foldl f z [] = z \}}
= ((1 * 2) * 3) * 4
  \remark{\{ 1 * 2 = 2 \}}
= (2 * 3) * 4
  \remark{\{ 2 * 3 = 6 \}}
= 6 * 4
  \remark{\{ 6 * 4 = 24 \}}
= 24
\end{Verbatim}
\end{minipage}
\begin{minipage}{0.45\textwidth}
\begin{Verbatim}[commandchars=\\\{\},commentchar=\%]
  foldl' (*) 1 [2, 3, 4]
  \remark{\{ foldl' f !z (x:xs) = foldl' f (f z x) xs \}}
= foldl' (*) (1 * 2) [3, 4]
  \remark{\{ 1 * 2 = 2 \}}
= .... 2
  \remark{\{ foldl' f !z (x:xs) = foldl' f (f z x) xs \}}
= foldl' (*) (2 * 3) [4]
  \remark{\{ 2 * 3 = 6 \}}
= .... 6
  \remark{\{ foldl' f !z (x:xs) = foldl' f (f z x) xs \}}
= foldl' (*) (6 * 4) []
  \remark{\{ 6 * 4 = 24 \}}
= .... 24
  \remark{\{ foldl' f !z [] = z \}}
= 24
\end{Verbatim}
\end{minipage}
\caption{Evaluation traces for \texttt{foldl} and \texttt{foldl'}.}
\label{fig:eval-foldl}
\end{figure*}

Bang patterns can be used to define a variant of the \texttt{foldl}
with a strict accumulator:
\begin{verbatim}
foldl' f z [] = z
foldl' f !z (x:xs) = foldl' f (f z x) xs
\end{verbatim}
Figure~\ref{fig:eval-foldl} show the traces for \verb|foldl| and
\verb|fold'| applied to the same arguments.  It is immediate that
\verb|foldl| accumulates an unevaluated expression while \verb|foldl'|
reduces the accumulator at each recursion step, thus avoiding a space
leak.

One pitfall that programmers need to be aware of is that \verb|foldl'|
forces evaluation of the accumulator \emph{only} to whnf; this means
that we can still get space leaks with lazy data types, e.g.\@ tuples.
Consider the function \verb|sumcount| that
computes the pair with the length and sum of a list of numbers:
\begin{verbatim}
sumcount = foldl' step (0,0)
step (n,s) x = (1+n,s+x)
\end{verbatim}
Although we are using the strict \verb|foldl'|, the accumulator is a
(lazy) pair, hence \texttt{sumcount} still exhibits space leaks.
This can solved by adding strictness
annotations in the components of the pair~\cite{RWH}:
\begin{verbatim}
sumcount = foldl' step (0,0)
step (!n,!s) x = (1+n,s+x)
\end{verbatim}
Figure~\ref{fig:sumcount} shows the evaluations traces for both
versions; observing the trace on the right, we can see that the
components of the pair are evaluated at each folding step, whereas the
trace on the left builds up two expressions that are evaluated after
the fold.  Hence, Haskelite traces can also be used to \emph{explain
  operational issues of lazy evaluation at a high-level}
without having to explain implementation concepts such as
thunks or examining GHC's Core or STG output.

\begin{figure*}
  \begin{minipage}{0.55\textwidth}
\begin{Verbatim}[commandchars=\\\{\},commentchar=\%]
  sumcount [1, 2, 3]
  \remark{\{ sumcount = foldl' step (0,0) \}}
= foldl' step (0, 0) [1, 2, 3]
  \remark{\{ foldl' f !z (x:xs) = foldl' f (f z x) xs \}}
= foldl' step (step (0, 0) 1) [2, 3]
  \remark{\{ step (n,s) x = (1+n,x+s) \}}
= ... (1 + 0, 1 + 0)
  \remark{\{ foldl' f !z (x:xs) = foldl' f (f z x) xs \}}
= foldl' step (step (1 + 0, 1 + 0) 2) [3]
  \remark{\{ step (n,s) x = (1+n,x+s) \}}
= ... (1 + (1 + 0), 2 + (1 + 0))
  \remark{\{ foldl' f !z (x:xs) = foldl' f (f z x) xs \}}
= foldl' step (step (1 + (1 + 0), 2 + (1 + 0)) 3) []
  \remark{\{ foldl' f z [] = z \}}
= step (1 + (1 + 0), 2 + (1 + 0)) 3
  \remark{\{ step (n,s) x = (1+n,x+s) \}}
= (1 + (1 + (1 + 0)), 3 + (2 + (1 + 0)))
  \remark{\{ 1 + 0 = 1 \}}
= (1 + (1 + 1), 3 + (2 + (1 + 0)))
  \remark{\{ 1 + 1 = 2 \}}
= (1 + 2, 3 + (2 + (1 + 0)))
  \remark{\{ 1 + 2 = 3 \}}
= (3, 3 + (2 + (1 + 0)))
  \remark{\{ 1 + 0 = 1 \}}
= (3, 3 + (2 + 1))
  \remark{\{ 2 + 1 = 3 \}}
= (3, 3 + 3)
  \remark{\{ 3 + 3 = 6 \}}
= (3, 6)
\end{Verbatim}
  \end{minipage}  
\begin{minipage}{0.4\textwidth}
\begin{Verbatim}[commandchars=\\\{\},commentchar=\%]
  sumcount [1, 2, 3]
  \remark{\{ sumcount = foldl' step (0,0) \}}
= foldl' step (0, 0) [1, 2, 3]
  \remark{\{ foldl' f !z (x:xs) = foldl' f (f z x) xs \}}
= foldl' step (step (0, 0) 1) [2, 3]
  \remark{\{ step (!n,!s) x = (1+n,x+s) \}}
= ... (1 + 0, 1 + 0)
  \remark{\{ foldl' f !z (x:xs) = foldl' f (f z x) xs \}}
= foldl' step (step (1 + 0, 1 + 0) 2) [3]
  \remark{\{ 1 + 0 = 1 \}}
= ....... 1
  \remark{\{ 1 + 0 = 1 \}}
= ....... 1
  \remark{\{ step (!n,!s) x = (1+n,x+s) \}}
= ... (1 + 1, 2 + 1)
  \remark{\{ foldl' f !z (x:xs) = foldl' f (f z x) xs \}}
= foldl' step (step (1 + 1, 2 + 1) 3) []
  \remark{\{ foldl' f z [] = z \}}
= step (1 + 1, 2 + 1) 3
  \remark{\{ 1 + 1 = 2 \}}
= ..... 2
  \remark{\{ 2 + 1 = 3 \}}
= ..... 3
  \remark{\{ step (!n,!s) x = (1+n,x+s) \}}
= (1 + 2, 3 + 3)
  \remark{\{ 1 + 2 = 3 \}}
= (3, 3 + 3)
  \remark{\{ 3 + 3 = 6 \}}
= (3, 6)
\end{Verbatim}
  \end{minipage}
  \caption{Evaluation traces for two versions of \texttt{sumcount}.}\label{fig:sumcount}
\end{figure*}

%% file: related.tex
\section{Related Work}\label{sec:related}

The basis for our work is the pattern-matching calculus of
Kahl~\cite{kahl_2004}.  The principle differences are: (1) we add a
\awhere\ binding construct to name sub-expressions; (2) we define an
operational semantics and normal forms for a call-by-need evaluation
strategy; (3) we do not consider the ``empty expression''
corresponding to a pattern matching failure $\lambda\matchfail$; there
is simply no evaluation in such cases.  The definitions of matching
arity and whnf's in Section~\ref{sec:bigstep} are also novel.

Chapter 4 of the classic textbook by Peyton~Jones~\cite{spj_1987}
defines the semantics of pattern matching using lambda abstractions
with patterns $\lambda p.E$ together with a \textsf{FAIL} expression
and a ``fatbar'' operator. This semantics is
denotational and serves primarily as the basis for defining the
correctness of compilation into case expressions presented in
the subsequent chapter.

Mitchell and Runciman have proposed the \emph{Catch} static analyser
for pattern matching safety in Haskell~\cite{mitchell_runciman_2005,
  mitchell_runciman_2008}, namely checking that pattern matchings are
\emph{exhaustive} and \emph{non-redundant}.  The analysis is defined
by translating the source program into a first-order core language
with only simple case expressions solving a set of generated
constraints.  This work is therefore orthogonal to the
contribution presented in this paper.

\emph{Hat} is a suite of tools for transforming Haskell~98 programs to
generate runtime traces and inspecting the resulting
traces~\cite{chitil_runciman_wallace_2003,hat_site} in various ways,
including for debugging or program comprehension.  Unlike Haskelite,
Hat supports full Haskell~98 and produces traces as an output file, which
works only for terminating programs (or at least programs that were
interrupted by the user).  The \emph{hat-observe} tool allows
inspecting the arguments and results of top-level functions, similarly
to the traces Haskelite produces. However, it does not highlight which
equations were used.

There are many teaching tools for the evaluation of computer programs.
One of the most popular is \emph{Python tutor}, a website
that allows visualizing the execution of Python, JavaScript, C, C++ and
Java programs~\cite{guo_2013,python_tutor}. The computational model is
strictly imperative: the program state is visualized as a pointer to the
current instruction and the current values of variables in scope.

In the functional programming community there is a long tradition of
teaching languages based on Scheme;
\emph{DrScheme} (now \emph{DrRacket}) is an IDE for
programming used for teaching which includes a graphical debugger.
This allows setting breakpoints, inspecting variables and step-by-step
execution. There is even a web-based version~\cite{WeScheme}.
However, Scheme does not encourage reasoning by pattern-matching
and equations that we are interested in teaching~\cite{wadler_1987_2}.

\emph{Hazel} is a live-programming environment based on a dialect of
the Elm language designed for teaching~\cite{HazelnutLive19}. It is
part of a larger research goal on contextual editors based on
\emph{typed holes}, where editing states of incomplete programs are
given a formal meaning~\cite{OmarVHSGAH17}; the subsequent
publication~\cite{Hazel_OOPSLA2023} presents a small-step semantics
for evaluation which preserves the ordering among possibly overlapping
patterns (this is necessary to be able to give a semantics to a
program with pattern holes).  However, the Hazel language is strict
rather than lazy, which simplifies the operational semantics but also
means that the technique does not transfer to the lazy setting.

More closely related to our work, Olmer et al.\@ have developed a
step-by-step evaluator for a subset of Haskell used in a tutoring
environment~\cite{olmer_evaluating_2014}.  The evaluator is based on a
general tool for defining rewriting systems and supports different
strategies (innermost or outermost); it can also be used to check
student's traces against an expected strategy.  However, unlike our
operational semantics and interpreter, it does not implement a lazy
evaluation strategy and does not support guards. Additionally, students
are unable to provide their own function definitions.


%% file: conclusion.tex
\section{Conclusion and Further Work}\label{sec:conclusion}

This paper presented a tracing
interpreter based on \lambdaPMC, a small lazy functional language
based on a pattern matching calculus together with two operational
semantics: a big-step semantics in the style of Launchbury and an
abstract machine in the style of Sestoft.  The principal contribution
of \lambdaPMC\ is that translating a source language such as Haskell
is more direct than with case expressions: each equation corresponds
to one alternative in a matching abstraction.  Furthermore, we have
shown that \lambdaPMC\ can seamlessly handle extensions such bang
patterns.  Finally, we have described an implementation of a tracing
interpreter based on this abstract machine
used for teaching a subset of Haskell.

A number of directions for extending this work are possible:

\paragraph{Binding patterns}
\lambdaPMC\ does not allow
patterns to occur in let or where bindings.
It may appear that we could re-use matchings
to translate a pattern binding:
\begin{equation}
  \llet{p=e_1}{e_2} ~\equiv~
    \lambda (e_1 \amatcharg (p \amatchpat e_2))  
  \end{equation}
  However, the above translation would not work if the binding is
  recursive (i.e.\@ the pattern variables of $p$ occur in $e_1$).
  Moreover, it does not respect the Haskell semantics: pattern
  bindings should introduce \emph{irrefutable
    patterns}~\cite{haskell_2010_report}; this means that the matching
  should be delayed until the pattern variables are used.  One
  approach for dealing with this is to employ the translation of
  irrefutable patterns into simple patterns as defined in the Haskell
  report.  Alternatively, we could investigate adding irrefutable
  patterns explicitly to \lambdaPMC.

  \paragraph{List comprehensions}
  We have not included support for list comprehensions in our
  interpreter; this is because the semantics for comprehensions is
  defined by translation into higher-order
  functions~\cite{haskell_2010_report} instead of an equational
  theory.  One alternative would be to perform the translations as part
  of the rewriting steps. However, this may be confusing to students, as
  comprehensions are sometimes introduced before higher-order
  functions~\cite{hutton_2016}.

  \paragraph{I/O and effects}
  We have only so far considered tracing purely functional programs.
  However, we do cover I/O in our introductory functional programming
  course.
  It would also be interesting to able to show
  evaluation of I/O code. One of the challenges is to balance
  simplicity with abstraction: should we de-sugar do-notation into
  monadic operations or keep it as a kind of pseudo-imperative layer?
  More experimentation is required to find the most effective
  approaches.
  
  \paragraph{Improving the interpreter UI} 
  The user interface in our implementation is currently rather basic;
  in particular, there are no mechanisms for skipping evaluation
  steps, controlling the granularity of evaluations or the depth of
  pretty-printing.  Being able to control evaluation directly from the
  UI might also be useful (for example, introducing \texttt{force} for
  strict evaluation).  Additional feedback from students could be
  helpful here.

  \paragraph{Educational assessment}
We are currently teaching an introductory functional programming
course where students are using Haskelite while learning.
Feedback from this experiment will be helpful to understand how
useful is the interpreter and possible improvements
to be made.
